\documentclass[a4paper,11pt]{article}
\usepackage[utf8x]{inputenc}
\usepackage{fullpage}
\usepackage{amsthm,amsmath, amssymb}
\usepackage{natbib}
\usepackage{ifpdf}
\ifpdf
  \usepackage[pdftex]{hyperref}
\else
  \usepackage[hypertex]{hyperref}
\fi

\theoremstyle{plain}
\newtheorem{theorem}{Theorem}
\newtheorem{lemma}{Lemma}
\newtheorem{claim}{Claim}
\newtheorem{cor}{Corollary}
\newtheorem{obs}{Observation}

\newcommand{\out}{\mbox{$\mathcal{O}$}}
\newcommand{\real}{\mbox{$\mathbb{R}$}}
\newcommand{\M}{\mathcal{M}}

\newcommand{\T}{\mathcal{T}}

\newcommand{\vect}{\mathbf }
\newcommand{\x}{\vect x}
\newcommand{\y}{\vect y}
\newcommand{\p}{\vect p}
\renewcommand{\b}{\vect b}

\newcommand{\s}{\vect s}
\newcommand{\bi}{{\vect{b}}_{-i}}

\title{Obvious Strategyproofness Needs Monitoring for Good Approximations}

\author{Diodato Ferraioli\thanks{DIEM, Universit\`a degli Studi di Salerno, Italy, \texttt{dferraioli@unisa.it}.
The author was supported by ``GNCS-INdAM''.} \and
Carmine Ventre\thanks{CSEE, University of Essex, UK, \texttt{c.ventre@essex.ac.uk}. The author was supported by the EPSRC grant EP/M018113/1.}}
\date{}

\begin{document}

\allowdisplaybreaks

\raggedbottom

\maketitle

\begin{abstract}
Obvious strategyproofness (OSP) is an appealing concept as it allows
to maintain incentive compatibility even in the presence of agents
that are not fully rational, e.g., those who struggle with contingent
reasoning \citep{li2015obviously}. However, it has been shown to impose some
limitations, e.g., no OSP mechanism can return a stable matching \citep{ashlagi2015no}.

We here deepen the study of the limitations of OSP mechanisms by
looking at their approximation guarantees for basic optimization
problems paradigmatic of the area, i.e., machine scheduling and
facility location. We prove a number of bounds on the approximation
guarantee of OSP mechanisms, which show that OSP can come at a
significant cost. However, rather surprisingly, we prove that OSP
mechanisms can return optimal solutions when they use monitoring — a
novel mechanism design paradigm that introduces a mild level of
scrutiny on agents’ declarations~\citep{kovacs2015mechanisms}.
\end{abstract}

\section{Introduction}
Algorithmic Mechanism Design (AMD) is by now an established research area in computer science that aims at conceiving algorithms resistant to selfish manipulations. As the number of parties (a.k.a., agents) involved in the computation increases, there is, in fact, the need to realign their individual interests with the designer's. Truthfulness is the chief concept to achieve that: in a truthful mechanism, no selfish and \emph{rational} agent has an interest to misguide the mechanism. A valid question of recent interest is, however, how easy it is for the selfish agents to understand that it is useless (and possibly costly) to strategize against the truthful mechanism at hand. 

Recent research has come up with different approaches to deal with this question. Some authors \citep{sandholm2003sequences,chawla2010multi,babaioff2014simple,adamczyk2015sequential} suggest to focus on ``simple'' mechanisms; e.g., in posted-price mechanisms one's own bid is immaterial for the price paid to get some goods of interest -- this should immediately suggest that trying to play the mechanism is worthless no matter the cognitive abilities of the agents.  However, in such a body of work, this property remains unsatisfactorily vague. An orthogonal approach is that of verifiably truthful mechanisms \citep{BP15}, wherein agents can run some algorithm to effectively check that the mechanism is incentive compatible. Nevertheless, these verification algorithms can run for long (i.e., time exponential in the input size) and are so far known only for quite limited scenarios. Importantly, moreover, they seem to transfer the question from the mechanism itself to the verification algorithm.

\citet{li2015obviously} has recently formalized the aforementioned idea of simple mechanisms, by introducing the concept of Obviously Strategy-Proof (OSP) mechanisms. This notion stems from the observation that the very same mechanism can be more or less truthful in practice depending on the implementation details. For example, in lab experiments, Vickrey's famous second-price mechanism results to be ``less'' truthful when implemented via a sealed-bid auction, and ``more'' truthful when run via an ascending auction. The quite technical definition of OSP formally captures how implementation details matter by looking at a mechanism as an extensive-form game; roughly speaking, OSP demands that strategyproofness holds among subtrees of the game (see below for a formal definition). An important validation for the `obviousness' is further 
provided by \citet{li2015obviously} via a characterization of these mechanisms in terms of agents with limited cognitive abilities (i.e., agents with limited skills in contingent reasoning). Specifically, Li shows that a strategy is obviously dominant if and only if these ``limited'' agents can recognize it as such. OSP is consequently a very appealing notion as in many cases rationality has been seen as the main obstacle to concrete applications of mechanism design paradigms, cf. e.g. \citep{FV15}; such a relaxation might be a panacea in these cases.

Nevertheless, for all its significant aspects, there appear to be  hints that the notion of OSP mechanisms might be too restrictive. \citet{ashlagi2015no} prove, for example, that no OSP mechanism can return a stable matching -- thus implying that the Gale-Shapley matching algorithm is not OSP despite its apparent simplicity.

\subsection{Our Contribution} 
We investigate the power of OSP mechanisms in more detail from a theoretical computer science perspective.
In particular, we want to understand the quality of approximate solutions that can be output by OSP mechanisms.
To answer this question, we focus on two fundamental optimization problems, machine scheduling \citep{AT01}
and facility location \citep{Mou80}, arguably (among) the paradigmatic problems in AMD.

For the former problem, we want to compute a schedule of jobs on selfish related machines (i.e., machines with job-independent speeds) so to minimize the makespan. For this single-dimensional problem, it is known that a truthful PTAS is possible \citep{CK13}. In contrast, we show that there is no better than $2$-approximate OSP mechanism for this problem independently from the running time of the mechanism. 

For the facility location problem, we want to determine the location of a facility on the real line given the preferred locations of $n$ agents. The objective is to minimize the social cost, defined as the sum over the individual agents of the distances between their preferred location and the facility's. \citet{Mou80} proves that the optimal mechanism, that places the facility on the median of the reported locations, is truthful without money (i.e., the mechanism does not pay or charge the agents). OSP mechanisms without money turn out to be much weaker than that. We prove in fact a tight bound of $n-1$. Interestingly, this bound can be shown also for mechanisms that use money, thus showing that transfers are not useful at all to enforce OSP. The proof of this fact uses a novel lower bounding technique for OSP mechanisms wherein the bidding domain (or, equivalently the strategy set) of the lying agent does not necessarily have size two (whereas, to the best of our knowledge, all previous impossibility results rely on this assumption); we in fact show that it is enough to identify two particular values in the bidding domain for our argument to work no matter the size of the domain. 

However, a surprising connection of OSP mechanisms with a novel mechanism design para\-digm -- called \emph{monitoring} -- allows us to prove strong positive results. Building upon the notion of mechanisms with verification \citep{NisRon99,Ven14,PenVen09}, \citet{kovacs2015mechanisms} introduce the idea that a mechanism can check the declarations of the agents at running time and guarantee that those who overreported their costs end up paying the exaggerated costs. This can be enforced whenever costs can be easily measured and certified. For example, a mechanism can force a machine that in her declaration has augmented her running time to work that long by keeping her idle for the difference between real and reported running time. 

We first prove that, no matter the algorithm at hand, there exists an OSP mechanism with monitoring that always compensates the agents their costs. This first-price mechanism, introduced in the context of mechanisms with monitoring by \citet{SVV16}, yields a couple of interesting observations in the context of OSP mechanisms. Firstly, it is the first direct-revelation OSP mechanism. As such it does not need any assumption on the agents' bidding domain nor to repeatdly query/interact with the agents.
We remark that the known constructions of \cite{li2015obviously} typically need finite domains (or, slightly more generally, domains admitting a finite partition).  Secondly, the mechanism uses the algorithm at hand as a black box thus reconciling approximation and obvious strategyproofness. This is relavant because, even for truthful mechanism, it is known that, without monitoring, black box reduction to a non-truthful algorithm may not exist \citep{dobzinski2012computational}. This theorem can then be applied to both our problems of interest, and prove the existence of optimal OSP mechanisms with monitoring. Clearly, the optimal mechanism for machine scheduling runs in exponential time; a PTAS that is OSP can however be obtained by plugging in the appropriate approximation algorithm.

Nevertheless, as noted in related literature \citep{PT09}, paying the agents in facility location problems might not be feasible in certain contexts. We therefore look at mechanisms that charge agents to use the facilities (e.g., via a subscription fee); note that, as stated above, transfers ought to be used for good approximations. We design the \emph{interval mechanism} for facility location that is optimal, OSP with monitoring and charges (rather than paying) the agents. This construction adapts the first-price mechanism in \citep{SVV16} to guarantee obvious strategyproofness even in absence of funds to pay the agents. We further show that this simple adaptation can be cleverly modified to guarantee that agents are charged as little as possible; we in fact prove that in general there is no OSP mechanism that is more economical than that. 

Our results for facility location draw an interesting parallel between OSP and truthful mechanisms. On one hand, our bounds for OSP mechanisms without money are in stark contrast with the case of strategyproof mechanisms where the optimum is known to be truthful \citep{Mou80}. On the other hand, our results for OSP mechanisms can be likened to truthful mechanisms for $K$-facility location, $K\geq 2$, where there is a linear gap between truthful approximations with and without money \citep{PT09,FT14} (incidentally, there are hints that the gap remains linear also for relaxed notions of truthfulness without money \citep{aamas16}) -- in the case of OSP mechanisms the price to pay to close the gap is not only money, but also monitoring.

\section{Preliminaries}
\smallskip \noindent {\bf Mechanisms and Strategyproofness.}
In this work we consider a classical mechanism design setting, in which we have a set of outcomes $\out$ and $n$ selfish agents.
Each agent $i$ has a \emph{type} $t_i \in D_i$, where $D_i$ is defined as the \emph{domain} of $i$.
The type $t_i$ is \emph{private knowledge} of agent $i$.
Moreover, each selfish agent $i$ has a \emph{cost function} $c_i \colon D_i \times \out \rightarrow \real$.
For $t_i \in D_i$ and $X \in \out$, $c_i(t_i, X)$ is the cost paid by agent $i$ to implement $X$ when her type is $t_i$.

A \emph{mechanism} consists of a protocol whose goal is to determine an outcome $X \in \out$.
To this aim, the mechanism is allowed to interact with agents.
During this interaction, agent $i$ is observed to take \emph{actions} (e.g., saying yes/no);
these actions may depend on her presumed type $b_i \in D_i$ that can be different from the real type $t_i$
(e.g., saying yes could ``signal'' that the presumed type has some properties that $b_i$ alone might enjoy).
We say that agent $i$ takes \emph{actions according to $b_i$} to stress this.
For a mechanism $\M$, we let $\M(\b)$ denote the outcome returned by the mechanism
when agents take actions according to their presumed types $\b = (b_1, \ldots, b_n)$.
Usually, this outcome is given by a pair $(f,\p)$,
where $f = f(\b)$ (termed \emph{social choice function} or, simply,  algorithm) maps the actions taken by the agents according to $\b$
to a feasible solution for the problem at the hand (e.g., an allocation of jobs to machines that enjoys particular properties),
and $\p=\p(\b)=(p_1(\b),\ldots,p_n(\b)) \in \real^n$ maps the actions taken by the agents according to $\b$
to \emph{payments} from the mechanism to each agent $i$.
Note that the $p_i$'s can be positive (meaning that the mechanism will pay the agents) or negative (meaning that the agents will pay the mechanism).

A mechanisms is said \emph{without money} if $p_i(\b)=0$ for every agent $i$ and every profile $\b \in D = D_1 \times \cdots \times D_n$.
Our definitions below do naturally extend to this case by considering null payments.

A mechanism $\M$ is \emph{strategy-proof} if for every $i$, every $\bi=(b_1, \ldots, b_{i-1}, b_{i+1}, \ldots, b_n)$
and every $b_i \in D_i$, it holds that $c_i(t_i,\M(t_i, \bi)) \leq c_i(t_i,\M(b_i, \bi))$,
where $t_i$ is the true type of $i$.
That is, in a strategy-proof mechanism the actions taken according to the true type are dominant for each agent.

Moreover, a mechanism $\M$ is said to satisfy \emph{voluntary participation} if for every $i$ and every $\bi$,
it holds that $c_i(t_i,\M(t_i, \bi)) \leq 0$.

\paragraph{Obvious Strategyproofness.}
Let us now formally define the concept of obviously strategy-proof mechanism.
This concept has been introduced in~\citep{li2015obviously}.
The original definition turns out to be very general and, consequently, quite complex.
For this reason, in this work we follow~\citet{ashlagi2015no} and rephrase this definition for our setting of interest.
Note that we focus on deterministic mechanisms only.

We begin by formally modeling how a mechanism works and subsequently  give some intuition behind the mathematical definition.
Specifically, we have that an \emph{extensive-form mechanism} $\M$ is defined by a directed tree $\T=(V,E)$ such that:
\begin{itemize}
 \item every leaf $\ell$ of the tree is labeled by a possible outcome $X(\ell) \in \out$ of the mechanism;
 \item every internal vertex $u \in V$ is labeled by a subset $S(u) \subseteq [n]$ of agents;
 \item every edge $e=(u,v) \in E$ is labeled by a subset $T(e) \subseteq D$ of type profiles such that:
 \begin{itemize}
  \item the subsets of profiles that label the edges outgoing from the same vertex $u$ are disjoint,
  i.e., for every triple of vertices $u, v, v'$ such that $(u,v) \in E$ and $(u,v') \in E$, we have that $T(u,v) \cap T(u,v') = \emptyset$;
  \item the union of the subsets of profiles that label the edges outgoing from a non-root vertex $u$
  is equal to the subset of profiles that label the edge going in $u$,
  i.e., $\bigcup_{v \colon (u,v) \in E} T(u,v) = T(\phi(u),u)$, where $\phi(u)$ is the parent of $u$ in $\T$; 
  \item the union of the subsets of profiles that label the edges outgoing from the root vertex $r$
  is equal to the set of all profiles, i.e., $\bigcup_{v \colon (r,v) \in E} T(r,v) = D$;
  \item for every $u, v$ such that $(u,v) \in E$ and for every two profiles $\b, \b' \in T(\phi(u),u)$ such that $(b_i)_{i \in S(u)} = (b'_i)_{i \in S(u)}$,
  if $\b$ belongs to $T(u,v)$, then also $\b'$ must belong to $T(u,v)$.
 \end{itemize}
\end{itemize}

Roughly speaking, the tree represents the steps of the execution of the mechanism.
As long as the current visited vertex $u$ is not a leaf, the mechanism concurrently interacts with agents in $S(u)$. 
Different edges outgoing from vertex $u$ are used for modeling the different actions that agents can take during this interaction with the mechanism.
In particular, each possible action is assigned to an edge outgoing from $u$.
As suggested above, the action that agent $i$ takes may depend on her presumed type $b_i \in D_i$.
That is, different presumed types may correspond to taking different actions,
and thus to different edges.
The label $T(e)$ on edge $e=(u,v)$ then lists the type profiles that enable agents in $S(u)$ to take those actions that have been assigned to $e$.
In other words, when the agents take the actions assigned to edge $e$,
then the mechanism (and the other agents) can infer that the type profile must be contained in $T(e)$.
The constraints on the edges' label can be then explained as follows:
first we can safely assume that different actions must correspond to different type profiles
(indeed, if two different actions are enabled by the same profiles we can consider them as a single action);
second, we can safely assume that each action must correspond to at least one type profile
that has not been excluded yet by actions taken before node $u$ was visited
(otherwise, we could have excluded this type profile earlier);
third, we have that the action taken by agents in $S(u)$ can only inform about types of agents in $S(u)$
and not about the type of the remaining agents (that are completely unknown to agents in $S(u)$).
The execution ends when we reach a leaf $\ell$ of the tree. In this case, the mechanism returns the outcome that labels $\ell$.

Observe that, according to the definition above, for every profile $\b$ there is only one leaf $\ell = \ell(\b)$
such that $\b$ belongs to $T(\phi(\ell),\ell)$. For this reason we say that $\M(\b) = X(\ell)$.
Moreover, for every type profile $\b$ and every node $u \in V$, we say that $\b$ is \emph{compatible} with $u$ if $\b \in T(\phi(u),u)$.
Finally, two profiles $\b$, $\b'$ are said to \emph{diverge} at vertex $u$ if there are two vertices $v, v'$
such that $(u,v) \in E$, $(u,v') \in E$ and $\b \in T(u,v)$, whereas $\b' \in T(u,v')$.

We are now ready to define obvious strategyproofness.
An extensive-form mechanism $\M$ is \emph{obviously strategy-proof (OSP)} if for every agent $i$,
for every vertex $u$ such that $i \in S(u)$, for every $\bi, \bi'$, and for every $b_i \in D_i$
such that $(t_i, \bi)$ and $(b_i, \bi')$ are compatible with $u$, but diverge at $u$,
it holds that $c_i(t_i, \M(t_i, \bi)) \leq c_i(t_i,\M(b_i, \bi'))$.
Roughly speaking, an obvious strategy-proof mechanism requires that, at each time step
agent $i$ is asked to take a decision that depends on her type, the worst cost that
she can pay if at this time step she behaves according to her true type
is at least the same as the best cost achievable by behaving as she had a different type.

Hence, if a mechanism is obviously strategy-proof, then it is also strategy-proof.
Indeed, the latter requires that truthful behavior is a dominant strategy when agents know the entire type profile,
whereas the former requires that it continues to be a dominant strategy
even if agents have only a partial knowledge of profiles\footnote{In fact,
OSP implies -- but is not equivalent to -- weakly group strategyproofness \citep{li2015obviously}.},
limited to what they observed in the mechanism up to the time they are called to take their choices. 

We say that an extensive-form mechanism is \emph{trivial} if for every vertex $u \in V$ and for every two type profiles $\b,\b'$,
it holds that $\b$ and $\b'$ do \emph{not} diverge at $u$. That is, a mechanism is trivial if it never requires that agents
take actions that depend on their type.
Observe that if a mechanism $\M$ is not trivial, then every path from the root to one leaf goes through a vertex $u^\star$ such that
there are two type profiles $\b, \b'$ that diverge at $u^\star$. Since $\b \neq \b'$, then there exists at least one agent $i^\star$ such that
$b_{i^\star} \neq b'_{i^\star}$. Moreover, by our definition of extensive-form mechanism, it must be the case that $i^\star \in S(u^\star)$.
For this reason, we call $i^\star$ as the \emph{divergent agent} for the mechanism $\M$.
Note that the divergent agent takes a decision that depends on her own type
before any other agents revealed any information about their own type.
For this reason, in order to prove that a mechanism is not obviously strategy-proof,
it is sufficient to show that there are two type profiles $\b, \b'$ with $b_{i^\star} \neq b'_{i^\star}$ such that
they diverge at $u^\star$, and $c_{i^\star}(b_{i^\star},\M(\b)) > c_{i^\star}(b_{i^\star},\M(\b'))$.

Let us state two further properties of obvious strategyproofness, that turn out to be very useful in the rest of the paper.
First, it is not hard to see that if $\M$ is OSP when the type profile is taken from $D$,
then it continues to be OSP even if the types are only allowed to be selected
from $D' = D'_1 \times \cdots \times D'_n$, where $D'_i \subseteq D_i$.
Moreover, let us define $\M'$ obtained from $\M$ by \emph{pruning} the paths involving actions corresponding to types in $D \setminus D'$.
If $\M$ is OSP, then also $\M'$ enjoys this property~\citep{li2015obviously}.

\paragraph{Monitoring.}
Let $\M(\b)$ denote the outcome returned by mechanism $\M=(f,\p)$ when agents take actions according to $\b$. Commonly, the cost paid by agent $i$ to implement $\M(\b)$ is defined as a quasi-linear combination of agent's true cost\footnote{Note that $t_i(f(\b))$ depends only on the type of the agent and the outcome of the social choice function.} $t_i(f(\b))$
and payment $p_i(\b)$, i.e., $c_i(t_i, \M(\b)) = t_i(f(\b)) - p_i(\b)$.
This approach disregards the agent's declaration for evaluating her cost.

In mechanisms with monitoring the usual quasi-linear definition is maintained but costs paid by the agents are more strictly tied to their declarations \citep{kovacs2015mechanisms}.
Specifically, in a {mechanism with monitoring} $\M$, the bid $b_i$ is a lower bound on agent $i$'s cost for $f(b_i, \bi)$, so an agent is allowed to have a real cost higher than $b_i(f(\b))$ but not lower.\footnote{We highlight that the designer only checks that agents are not ``faster'' than declared. That is, agents can pretend to have a higher cost/processing time at the expense of being ``busy'' that long (e.g., designer and agents could be in the same room). Agents can still underbid and at execution time have a higher cost (e.g., they could say to have underestimated their cost/work). Note that contrarily to the notion of verification in \citep{NisRon99} there is here no punishment for this misbehavior.}
Formally, we have \[
c_i(t_i,\M(\b)) = \max\{t_i(f(\b)), b_i(f(\b))\} - p_i(\b).
\]

\smallskip \noindent We next describe two specific problems of interest.

\paragraph{Machine Scheduling.}
\smallskip 
Here, we are given a set of $m$ different jobs to execute and the $n$ agents control related machines.
That is, agent $i$ has a job-independent processing time $t_i$ per unit of job
(equivalently, an execution speed $1/t_i$ that is independent from the actual jobs).
The social choice function $f$ must choose a possible schedule $f(\b) = (f_1(\b), \ldots, f_n(\b))$ of jobs to the machines,
where $f_i(\b)$ denotes the job load assigned to machine $i$ when agents take actions according to $\b$.
The cost that agent $i$ faces for the schedule $f(\b)$ is $t_i(f(\b))=t_i \cdot f_i(\b)$. 
Note that our mechanisms for machine scheduling will always pay the agents.

Monitoring can be readily implemented for this setting.
In fact, monitoring  means that those agents who have exaggerated their unitary processing time,
i.e., they take actions according to $b_i>t_i$,
can be made to process up to time $b_i$ instead of the true processing time $t_i$.
For example, we could not allow any other operation in the time interval $[t_i, b_i]$ or charge $b_i-t_i$.

We focus on social choice functions $f^*$ optimizing the \emph{makespan}, i.e., 
\[
f^*(\b) \in \arg\min_{\x} \mathtt{mc}(\x, \b), \quad \mathtt{mc}(\x, \b) = \max_{i=1}^n b_i(\x).
\]
We say that $f$ is $\alpha$-approximate if it returns a solution whose cost is a factor $\alpha$ away from the optimum.

\paragraph{Facility Location.}
In the facility location problem, the type $t_i$ of each agent consists of her position on the real line.
The social choice function $f$ must choose a position $f(\b) \in \real$ for the facility.
The cost that agent $i$ pays for a chosen position $f(\b)$ is $t_i(f(\b)) = d(t_i,f(\b)) = |t_i - f(\b)|$. 
So, $t_i(f(\b))$ denotes the distance between $t_i$ and the location of the facility computed by $f$ when agents take actions according to $\b$.

We can implement monitoring also in this setting whenever evidences of the distance can be provided (and cannot be counterfeited).
In fact, in this context, monitoring means that $t_i(f(\b)) = \max \{d(t_i, f(\b)), d(b_i, f(\b))\}$.
Therefore, once the evidence is provided, the mechanism can check whether $t_i(f(\b)) < b_i(f(\b))$
and charge the agent the difference for cheating.\footnote{One relevant applicative scenario here is for example reimbursement of previously declared expenses. These expenses are usually reimbursed only upon production of receipts so that for agents to be consistent with overbidding they need to pay the exaggerated (reported) cost. Receipts are then a tool for ``monitoring'' the agents.}

We focus on social choice functions $f^*$ optimizing the \emph{social cost}, i.e., 
\[
f^*(\b) \in \arg\min_{x \in \real} \mathtt{cost}(x, \b), \quad \mathtt{cost}(x, \b) = \sum_{i=1}^n b_i(x).
\]
As above, we say that $f$ is $\alpha$-approximate if it returns a solution whose cost is at most a factor $\alpha$ away from the optimum.

\section{A General Positive Result}
For an algorithm $f$, define $p_i(\b)= b_i(f(\b))$. We call the direct-revelation mechanism $(f, p)$ a first-price mechanism.

\begin{theorem}\label{thm:general}
Any direct-revelation first-price mechanism is OSP with monitoring and satisfies voluntary participation.
\end{theorem}
\begin{proof}
In order to prove that $\M=(f,p)$ is OSP, consider agent $i$ and let $t_i$ be her true type. We next show that for agent $i$ it is always convenient to be truthful, regardless of the decisions taken by other agents.
To this aim, let us recall that in a mechanism with monitoring the cost that $i$ pays, given the submitted type profile is $\b$, is
$$c_i(t_i, \M(\b)) = \max\{t_i(f_i(\b)), b_i(f_i(\b))\}  - p_i(\b).$$
 
\noindent Suppose that $i$ is truthful; then for every $\b_{-i}$, it turns out that
$$c_i(t_i, \M(t_i,\b_{-i})) = t_i(f_i(t_i,\b_{-i})) - p_i(t_i,\b_{-i}) = 0.$$

\noindent Suppose, instead, that $i$ lies and says $b_i$. Then for all $\bi$, if $b_i(f_i(\b)) > t_i(f_i(\b))$ then
$$c_i(t_i, \M(b_i,\b_{-i})) = b_i(f_i(b_i,\b_{-i})) - p_i(b_i,\b_{-i}) = 0;$$ if, instead, $t_i(f_i(\b)) \geq b_i(f_i(\b))$ then $$c_i(t_i, \M(b_i,\b_{-i})) = t_i(f_i(b_i,\b_{-i})) - p_i(b_i,\b_{-i}) > 0.$$ 

Note that all the cost (in)equalities hold no matter the value of $\bi$. Thus, in both cases the best cost that $i$ can obtain by adopting a strategy different from the truthful one is not smaller than the worst cost that $i$ can obtain by adopting the truthful strategy, as desired.
\end{proof}

It is important to note that in the construction above, we may use every algorithm as a \emph{black box}. This in particular means that we can turn any optimal (approximation, resp.) algorithm into an optimal (approximate, resp.) OSP mechanism with monitoring (without losses to the approximation guarantee, resp.). Thus, for Combinatorial Auctions (CAs) with additive bidders our mechanism with monitoring beats the lower bound proved by \cite{BG16} for OSP mechanisms.\footnote{We presented our setting for agents having a cost to implement the solution chosen by the mechanism; clearly, in CAs, agents have a non-negative valuation (i.e., non-positive cost) for the outcome of the auction. Since the theorem does not require the agents' costs to be positive, we can apply it also to CAs.} Just as the weaker notion of verification has been shown to be useful in the context of truthful CAs without money \citep{aamas14}, our result shows that OSP with monitoring matches the best-known (poynomial-time) approximations achieved not only  by truthful mechanisms \citep{jair15, LOS}, but also by general algorithms \citep{hs89, h99}. 

We also stress that this construction is \emph{query optimal}, as the interaction with each agent is minimum. It is worthy to observe that such an interaction does not need to be simultaneous (as assumed in practically all the literature on direct-revelation mechanisms) since obvious strategyproofness is maintained even if agents are queried in an adversarially chosen order and know what the bidders preceding them have declared. We will see how to exploit this property to reduce our payments for facility location. Finally, as observed above, we do not require the domain of each agent to be finite. 

\section{Machine Scheduling}
We now show that, without monitoring, there is no OSP mechanism that satisfies voluntary participation and returns an assignment of jobs to machines
whose makespan is at most twice the makespan of the optimal assignment.
Interestingly, this is the same lower bound that \citet{NisRon99} proved for the approximation ratio of strategy-proof mechanisms for \emph{unrelated} machines,
i.e., when it is not possible to express the processing time of jobs on machines as a product of jobs' load and machine's unit processing time.
We wonder if a more deep relationship exists between OSP mechanisms for scheduling on related machines
and strategy-proof mechanisms for scheduling on unrelated machines,
and if one can improve the lower bound for the former problem in order to match the best known lower bounds for the latter, i.e., $1 + \phi \approx 2.61$ for general mechanisms~\citep{koutsoupias2007lower}, and $n$ for anonymous mechanisms~\citep{ashlagi2012optimal}.
\begin{theorem}
 For every $\varepsilon > 0$, there is no $(2-\varepsilon)$-approximate mechanism for the machine scheduling problem
 that is OSP without monitoring and satisfies voluntary participation.
\end{theorem}
\begin{proof}
Let us consider the simple setting in which there are exactly two machines, that we denote with $0$ and $1$, and two equivalent jobs of unit length.
We will denote with $t_0$ and $t_1$ the type, i.e., the job processing time, of machine $0$ and $1$, respectively.
Suppose there is a $k$-approximate, with $k < 2$, OSP mechanism $\M$ that satisfies voluntary participation.

Since the mechanism is $k$-approximate, then it must be the case that:
if $t_0 < \frac{t_1}{2k}$, then $\M$ assigns both jobs to machine $0$;
if $t_0 > 2k \cdot t_1$, then $\M$ assigns both jobs to machine $1$;
if $\frac{k}{2} \cdot t_1 < t_0 < \frac{2}{k} \cdot t_1$, then $\M$ assigns one job to each machine.

Moreover, since mechanism $\M=(f,\p)$ is OSP, then it must be also strategy-proof.
\citet{AT01} proved that a mechanism for the machine scheduling problem is strategy-proof and satisfies voluntary participation
if and only if (i) the allocation of jobs to machine $i \in \{0,1\}$ returned by $f$ when the type of the other machine is $t_{1-i}$ is \emph{monotone},
i.e., $f_i(t_i, t_{1-i}) \leq f_i(t'_i, t_{1-i})$ whenever $t_i > t'_i$;
(ii) the payment that the machine $i$ receives is
$$p_i(t_i, t_{1-i}) = t_i f_i(t_i,t_{1-i}) + \int_{t_i}^{\infty} f_i(x,t_{1-i}) dx.$$

In our setting, the monotonicity requirement implies that, for every $t_{1-i}$, there are
$t' \in \left[\frac{t_{1-i}}{2k}, \frac{k}{2} \cdot t_{1-i}\right]$ and $t'' \in \left[\frac{2}{k} \cdot t_{1-i}, 2k \cdot t_{1-i}\right]$, such that
machine $i$ is assigned both jobs if $t_i < t'$, only one job if $t' \leq t_i \leq t''$, and no jobs if $t_i > t''$.
Hence, $p_i(t_i,t_{1-i}) = t' + t''$ if $t_i < t'$, $p_i(t_i,t_{1-i}) = t''$ if $t' \leq t_i \leq t''$, and $p_i(t_i,t_{1-i}) = 0$ otherwise.

Let us now restrict the domain of the agents to $D' = \{a, b\}^2$, with $b > k^2 a$.
Let $\M'$ be the mechanism obtained by pruning $\M$ according to this restriction.
As stated above, $\M'$ must be an OSP mechanism.
Moreover, the approximation ratio of $\M'$ cannot be worse than the approximation ratio of $\M$.
Hence, $\M'$ cannot be trivial (indeed, a trivial mechanism would have approximation ratio worse than $k$).

Let $i$ be the divergent agent of $\M'$. Clearly, $a$ and $b$ are the types in which $i$ diverges.
Suppose that $t_i = a$.
If $i$ behaves according to $t_i$, then it may be the case that the other agent behaves according type $a$ too.
As showed above, in this case machine $i$ receives one job and payment $t'' \leq 2k a$. Hence, $c_i(a, \M(a,a)) \geq a - 2k a$.
Suppose instead that $i$ behaves as if her type was $b$.
It may be the case that the other agent behaves according type $b$ too.
Then, machine $i$ still receives one job and a payment $t'' \geq \frac{2}{k} \cdot b$.  Hence,
$$
c_i(a, \M(b,b)) \leq a - \frac{2}{k} \cdot b < a - 2k a = c_i(a, \M(a,a)),
$$
where we used that $b > k^2 a$. 
In words, the best cost paid by $i$ if she does not behave according to her true type can be lower than the worst cost she can pay if she behaves according to her true type.
Then, the mechanism $\M'$ is not OSP, contradicting our hypothesis.
\end{proof}

Since there is a PTAS for the allocation of jobs to related machines~\citep{hochbaum1988polynomial},
then we have the following corollary of Theorem \ref{thm:general}.
\begin{cor}
 There is an OSP mechanism with monitoring that computes the optimal scheduling of jobs to related machines (in exponential time).
 Moreover, there is an OSP mechanism with monitoring that is a PTAS for the same problem. Both mechanisms satisfy voluntary participation.
\end{cor}

\section{Facility Location}
We now show that, without monitoring, there is no OSP mechanism for the facility location problem with an approximation ratio better than $n-1$. To this aim, let us first state the following simple observation.
\begin{obs}
 \label{obs:no_fac_extreme}
 For every $\alpha, \beta$, with $\alpha < \beta$, no $k$-approximate mechanism $\M = (f, \p)$, with $k < n-1$, sets $f(\x) \leq \alpha$, if $x_i = \alpha$ and $x_j = \beta$ for every $j \neq i$, and $f(\y) \geq \beta$, if $y_i =\beta$ and $y_j = \alpha$ for every $j \neq i$.
\end{obs}

\noindent The observation will be proved below in a more general statement.

\begin{theorem}
 For every $\varepsilon > 0$, there is no $\left(n - 1 - \varepsilon\right)$-approximate mechanism for the facility location problem that is OSP without monitoring.
\end{theorem}
\begin{proof}
 Let $\M = (f,\p)$ be a $\left(n - 1 - \varepsilon\right)$-approximate mechanism that is OSP without monitoring.
 Let us restrict the domain of every agent to $D' = \{a, a+ \delta, \ldots, b-\delta, b\}$,
 where $\delta \leq \frac{\varepsilon}{n-2} \cdot \frac{b-a}{2}$.
 Let $\M'$ be the mechanism obtained by pruning $\M$ according to this restriction.
 As stated above, $\M'$ must be an OSP mechanism.
 Moreover, the approximation ratio of $\M'$ cannot be worse than the approximation ratio of $\M$.
 Hence, $\M'$ cannot be trivial, otherwise its approximation ratio would be unbounded.
 
 Then, let $i$ be the divergent agent of $\M'$.
 Note that, by definition of divergent agent, there must be two types $t_i, t'_i$ of agent $i$ such that $t'_i = t_i + \delta$
 and $i$ takes an action in $\M'$ when her type is $t_i$ that is different from the action taken when her type is $t'_i$.
 We denote as $c$ and $d$ the smallest $t_i$ and the largest $t'_i$, respectively, for which this occurs,
 i.e., $c$ is the smallest type in $D'$ such that $i$ diverges on $c$ and $c+\delta$,
 and $d$ is the largest type in $D'$ such that $i$ diverges on $d$ and $d-\delta$.
 
 Note that either $c < \frac{b+a}{2}$ or $d > \frac{b+a}{2}$.
 Indeed, if $c \geq \frac{b+a}{2}$, then $d \geq c+\delta > \frac{b+a}{2}$.
 In the rest of the proof we will assume that $c < \frac{a+b}{2}$.
 The proof for the case that $d > \frac{a+b}{2}$ simply requires to replace $c$ with $d$, $c+\delta$ with $d-\delta$, and $b$ with $a$,
 and invert the direction of the inequalities in the next claims.
 
The proof uses two profiles $\x$ and $\y$, that are defined as follows:
 \begin{itemize}
 \item $x_i = c+\delta$, and $x_k = c$ for every $k \neq i,$;
 \item $y_i = c$, and $y_k = b$ for every $k \neq i$.
 \end{itemize}
We begin by using OSP to relate payments and outcomes of the mechanism $\M'$ on input $\x$ and $\y$. Specifically, we note that if the real location of $i$ is $t_i = x_i = c+\delta$ then
 $c_i(t_i,\M'(\x)) = d(c+\delta, f(\x)) - p_i(\x)$, and
 $c_i(t_i,\M'(\y)) = d(c+\delta, f(\y)) - p_i(\y)$.
 Since $i$ diverges on $c$ and $c+\delta$ and $\M'$ is OSP, we have that
 $c_i(t_i,\M'(\x)) \leq c_i(t_i,\M'(\y))$. Hence, it follows that
 \begin{equation}
  \label{eq:cond1}
  p_i(\x) \geq p_i(\y) - d(c+\delta, f(\y)) + d(c+\delta, f(\x)).
 \end{equation}

 Suppose instead that the real location of $i$ is $t'_i = y_i = c$ then
 $c_i(t'_i,\M'(\x)) = d(c, f(\x)) - p_i(\x)$, and
 $c_i(t'_i,\M'(\y)) = d(c, f(\y)) - p_i(\y)$.
 As above, since $i$ diverges on $c$ and $c+\delta$ and $\M'$ is OSP, we have that
 $c_i(t'_i,\M'(\y)) \leq c_i(t'_i,\M'(\x))$. Hence, it follows that
 \begin{equation}
  \label{eq:cond2}
  p_i(\x) \leq p_i(\y) - d(c, f(\y)) + d(c, f(\x)).
 \end{equation}
  
 Therefore, in order to satisfy both \eqref{eq:cond1} and \eqref{eq:cond2}, we need that
 \begin{equation}
  \label{eq:cond_fin}
  d(c+\delta, f(\y)) - d(c, f(\y)) \geq d(c+\delta, f(\x)) - d(c, f(\x)).
 \end{equation}

 Using \eqref{eq:cond_fin} above, we first show that $f(\x)$ must be at most $c$ and then that $f(\y) \leq c + \delta$. Finally, we prove how this last fact contradicts the desired approximation ratio.
 
 Let us first show that $f(\x) \geq c$. Suppose instead that $f(\x) < c$. Since $f(\x) < c$, then the r.h.s. of \eqref{eq:cond_fin} is $\delta$. As for the l.h.s., we distinguish two cases.
  If $f(\y) \leq c+\delta$, then, since $f(\y) > c$ according to Observation~\ref{obs:no_fac_extreme},
  then we have $(c+\delta - f(\y)) - (f(\y) - c) = \delta - 2(f(\y) - c) < \delta$.
  If $f(\y) > c + \delta$, we have $(f(\y) - (c+\delta)) - (f(\y) - c) = - \delta$. Hence, in both cases we reach a contradiction.
  
  We now show that $f(\y) \leq c+\delta$. Assume by contradiction that $f(\y) > c+\delta$. 
%
%
  Since $f(\x) \geq c$, and $f(\x) < c+\delta$ by Observation~\ref{obs:no_fac_extreme}, we can rewrite \eqref{eq:cond_fin} as
  $$
   (f(\y) - (c+\delta)) - (f(\y) - c) \geq ((c+\delta) - f(\x)) - (f(\x) - c) \Rightarrow -\delta \geq \delta - 2 (f(\x) - c).
  $$
However, this is impossible since $f(\x) < c +\delta$.
 
 Finally, we prove that, given that $f(\y) \leq c+\delta$, then the mechanism is not $(n-1-\varepsilon)$-approximate.
 Indeed, since by Observation~\ref{obs:no_fac_extreme} $f(\y) > c$, the total cost of mechanism $\M'$ on input $\y$ is
 \begin{align*}
  (f(\y) - c) + (n-1)\left(b - f(\y)\right) & = (n-1)b -c - (n-2)f(\y)\\
  & \geq (n-1)(b-c) - (n-2) \delta\\
  & \geq (n-1)(b-c) - (n-2) \frac{\varepsilon}{n-2} \cdot \frac{b-a}{2} > (n-1-\varepsilon) (b-c),
 \end{align*}
 where we used that $b-c > b - \frac{b+a}{2} = \frac{b-a}{2}$.
 However, this is absurd, since $\M'$ is $(n-1-\varepsilon)$-approximate and
 the optimal mechanism on input $\y$ places the facility in $b$ and has total cost $b-c$.
\end{proof}

Next show that if we insist on mechanisms without money, then there is no OSP mechanism that can guarantee an approximation ratio better than $n-1$ even when the mechanism can use monitoring.
\begin{theorem}
\label{thm:no_approx_fl}
 For every $\varepsilon > 0$, there is no $(n-1-\varepsilon)$-approximate mechanism without money for the facility location problem that is OSP,
 even with monitoring.
\end{theorem}
In order to prove Theorem~\ref{thm:no_approx_fl}, we first need to state the following lemma,
that can be seen as a quantitative version of Observation~\ref{obs:no_fac_extreme}.
\begin{lemma}
\label{lem:cond_fl}
 Consider a type profile $\b$ such that $b_i = x$ for some $i$ and $b_j = x-\alpha$ for every $j \neq i$.
 Then $f(\b) \in \left[x - \alpha\left(1 + \frac{k-1}{n}\right), x - \alpha\left(1 - \frac{k-1}{n-2}\right)\right]$
 for every $k$-approximate mechanism.
\end{lemma}
\begin{proof}
 The optimal facility location for the given setting consists in placing the facility in position $x-\alpha$.
 The total cost in this case is $\alpha$.
 
 If $f(\b) < x - \alpha\left(1 + \frac{k-1}{n}\right)$, then the total cost is larger than
 $(n-1) \frac{(k-1)\alpha}{n} + \alpha + \frac{(k-1)\alpha}{n} = k\alpha$,
 thus no $k$-approximate mechanism can place the facility in $f(\b)$.
 Similarly, if $f(\b) > x - \alpha\left(1 - \frac{k-1}{n-2}\right)$, then the total cost is $(n-1) (f-x+\alpha) + x-f = (n-2)(f-x) + (n-1)\alpha > k\alpha$,
 thus no $k$-approximate mechanism can place the facility in $f(\b)$.
\end{proof}
We are now ready to prove Theorem~\ref{thm:no_approx_fl}.
\begin{proof}[Proof of Theorem~\ref{thm:no_approx_fl}]
 Suppose there is an OSP mechanism $\M$ that is $(n-1-\varepsilon)$-approximate.
 Clearly, the mechanism is non-trivial, otherwise its approximation ratio would be unbounded.
 Then, let $i$ be the divergent agent of $\M$, and let $x_i$ and $y_i$ be the types in which $i$ diverges. W.l.o.g., assume that $x_i > y_i$.
 Let $\lambda = 2\left(x_i - y_i\right)$ and $\alpha = \lambda \cdot \frac{n-2}{\varepsilon}$.
 Let $x_i$ be the truthful position of this agent.
 If $i$ plays truthfully, then she can face the setting in which the remaining $n-1$ agents are in position $x_i - \alpha$.
 By applying Lemma~\ref{lem:cond_fl} with $k = n-1-\varepsilon$ and $x = x_i$, we have that the distance of agent $i$ from the facility
 must be at least $x_i - x_i + \alpha\left(1 - \frac{n-2-\varepsilon}{n-2}\right) = \alpha \cdot \frac{\varepsilon}{n-2} = \lambda$.
 
 Suppose that instead $i$ plays as if her real location would be $y_i$.
 It may be then the case that the remaining $n-1$ agents are exactly in the same position.
 Then, any mechanism with bounded approximation must place the facility in $y_i = x_i - \frac{\lambda}{2}$.
 Recall that, with monitoring, the cost of agent $i$ must be taken as the maximum between the distance
 to the facility either from the real position or from the declared position.
 In this case, this is given by the former distance and it is $\frac{\lambda}{2} < \lambda$.
 Thus, the best cost paid by $i$ by not playing truthfully is lower than the worst cost that she can pay by playing truthfully.
 Then, the mechanism $\M$ is not OSP, contradicting our hypothesis.
\end{proof}

The bounds above are tight, since there is a $(n-1)$-approximate mechanism without money for the facility location problem that is OSP,
even without monitoring.
Consider, indeed, the dictatorship mechanism, in which only the dictator $i$ is queried for her position.
It is well-known that this mechanism is $(n-1)$-approximate. We next observe that it is also OSP.
Agent $i$ is the only agent that is involved in a decision and it is always better for her to reveal her real position $x_i$:
indeed, in this case the facility will be located exactly in her position and the cost of $i$ will be 0,
whereas by declaring a different position $x \neq x_i$ the cost will be $|x - x_i| > 0$.

\subsection{Optimal OSP Mechanisms for Facility Location}
Interestingly, monitoring gives an enormous power in this setting.
Indeed, since the optimal facility location is the median among the positions declared by agents, it follows from Theorem \ref{thm:general} that there is an OSP mechanism with monitoring that computes the optimal facility location in polynomial time. 

Recall that in this mechanism the agents receive a payment.
As noted in the introduction, however, for facility location problem we might need 
alternative mechanisms 
in which it is not the designer to pay agents, but the agents to pay the mechanism. Note that this is more natural in settings wherein agents' payments can be easily implemented via subscription fees or delayed access to the facility. We next present such an alternative optimal OSP mechanism.

We are going to assume that we are given some bounds on the agents' potential locations. (Note that in some of the related literature on facility location, agents can declare any location in $\real$.)
To simplify the notation, we assume that $D_i=[a,b]$ for all agents $i$. Consider now the following direct-revelation mechanism, that we call \emph{interval mechanism}:
 \begin{enumerate}
  \item Query agents for their position.
  \item Let $\x$ be the profile of the collected positions. Then fix the location $f(\x)$ of the facility to be the median of $\x$.
  In case of multiple medians, the facility is located on the leftmost median.
  \item For every agent $i = 1, \ldots, n$, set $p_i(\x)=d(x_i, f(\x)) - (b-a)$.
 \end{enumerate}

It is not hard to see that the interval mechanism simply ``shifts'' the payments of the mechanism in Theorem \ref{thm:general} to make them of the right sign. In a sense, the theorem below proves that, just like truthfulness, OSP is preserved when these shifts are bid independent. 

\begin{theorem}
The interval mechanism
is an optimal mechanism that is OSP with monitoring.
\end{theorem}
\begin{proof}
 We will next prove that the mechanism is OSP, and thus each agent has an incentive to declare her real position.
 Since the mechanism places the facility in the median of these positions, it then turns out to be optimal as well.
 
 In order to prove that it is OSP, 
 recall that in a mechanism with monitoring the cost that $i$ pays is
 $c_i(x_i, \M(\y)) = \max\{d(x_i,f(\y)), d(y_i,f(\y))\} - p_i(\y)$.
 Consider then agent $i$ and let $x_i$ be her real position.
 If $i$ declares the real position, then her total cost will be $b-a$.
 If $i$ declares a different position $x'_i$, then there are two cases:
 if $\min_{\x'_{-i}} c_i(x_i,\M(\x'))$ is achieved in a profile $\x'_{-i}$ such that $f(\x') \neq x'_i$, then
 \begin{align*}
  c_i(x_i, \M(\x')) &= \max\{d(x_i,f(\x')), d(x'_i,f(\x'))\} - p_i(\x')\\
  & \geq d(x'_i,f(\x')) - p_i(\x') = b-a;
 \end{align*}
 otherwise (that is, if $f(\x') = x'_i \neq x_i$)
 \begin{align*}
  c_i(x_i, \M(\x')) &= \max\{d(x_i,f(\x')), d(x'_i,f(\x'))\} - p_i(\x')\\
  & = d(x_i,x'_i) - p_i(\x') > b-a.
 \end{align*}
 Thus, in both cases the best cost that $i$ can obtain by declaring a position different from the real one
 is not smaller than the worst cost that $i$ can obtain by playing truthfully.
\end{proof}

The drawback of the interval mechanism is that the payment that this mechanism charges may be as large as the size of the interval. This opens the question of whether more frugal payment schemes exist --- or in other words, how susceptible OSP with monitoring is to payment shifts that are not bid independent. 

We will show in Section~\ref{apx:frugal} that it is indeed possible to slightly optimize the interval mechanism in order to be less expensive for the agents. We further prove that our optimization is optimal as long as we focus on direct-revelation mechanisms. However, even this optimized version still requires that $O(n)$ agents will pay an amount that is about $b-a$. We will show that this is somewhat unavoidable, even if one considers mechanisms that are not direct-revelation.

\subsubsection{The Optimized Interval Mechanism}\label{apx:frugal}
Consider the following optimized version of the interval mechanism, that we call \emph{Optimized Interval Mechanism (OIM)}:
 \begin{enumerate}
  \item Query agents for their position.
  \item Let $\x$ be the profile of the collected positions. Then fix the location $f(\x)$ of the facility to be the median of $\x$.
  In case of multiple medians, the facility is located on the leftmost median.  
  \item For every agent $i \in [n]$, let $K_i$ ($k_i$, resp.) be the set (number, resp.) of agents queried before $i$.
  Let $\s = (s_1, \ldots, s_{k_i})$ be the profile containing the locations declared by these agents in non-decreasing order. Let $\ell = \left\lceil \frac{n}{2} \right\rceil  + k_i - n + 1$ and $r = \left\lceil \frac{n}{2} \right\rceil - 1$.
  If $\ell>1$, $r< k_i$ and $s_{\ell - 1} = s_{r+1}$, then set $p_i(\x) = 0$ for every $\x=((\s, x_i), \x_{-K_i\cup\{i\}})$. 
  
  Otherwise, we define $L_i$ and $R_i$ as follows:
  $$L_i = \begin{cases}
           s_\ell, & \text{if } \ell \geq 1;\\
           a, & \text{otherwise}.
          \end{cases}
  \qquad
  R_i = \begin{cases}
           s_r, & \text{if } r \leq k_i;\\
           b, & \text{otherwise}.
          \end{cases}
  $$
  Let also define $A_i$ and $B_i$ as follows:
  $$A_i = \begin{cases}
	   L_i, & \text{if } L_i = R_i;\\
           2L_i - b, & \text{if } R_i > L_i > \frac{a+b}{2};\\
           a, & \text{otherwise}.
          \end{cases}
    \qquad
    B_i = \begin{cases}
	   R_i, & \text{if } L_i = R_i;\\
           2R_i - a, & \text{if } L_i < R_i < \frac{a+b}{2};\\
           b, & \text{otherwise}.
          \end{cases}
  $$
  Finally, let $m_i = \max \{R_i - A_i, B_i - L_i\}$.
  
  If $x_i \in [A_i, B_i]$,
  then set $p_i(\x) = d(x_i, f(\x)) - m_i$ for every $\x$.
  
  
  \noindent If $x_i < A_i$,
  then set $p_i(\x) = d(x_i, f(\x)) - m_i - d(x_i, A_i)$ for every $\x$.

  \noindent If $x_i > B_i$,
  then set $p_i(\x) = d(x_i, f(\x)) - m_i - d(x_i, B_i)$ for every $\x$.
 \end{enumerate}
 
The idea behind OIM is to exploit the information given by the interactive implementation of the mechanism to reduce the charge to the bidders, i.e., use the value of $k_i$ to reduce the payment to bidder $i$. In fact, when all bidders bid simultaneously then $k_i=0$ for all $i$ and OIM is simply the interval mechanism.

The way in which this optimization upon $k_i$ is done can arguably appear a bit complex but is not too hard to explain. First, when $s_{\ell-1} = s_{r+1}$, then the facility will be placed in $s_{\ell-1}$ regardless of the location
declared by $i$ and by every other agent $j$ queried after $i$. 
Hence, these agents will not have any incentive in declaring a position that is different from their real location even without payments.


As for the second and most important optimization step, we consider profiles $\x$ for which there are agents very far away from the facility. Indeed, as we will hint in Lemma~\ref{lem:propLR}, the facility is very likely to be included in the interval $[L_i, R_i]$. Thus, if an agent $i$ in $\x$ is very far away from this interval, one can slightly lower the payment assigned to her for every other profile and still have an OSP mechanism. Indeed, no such agent has an incentive to either move from another profile to $\x$ (obviously) nor to move from $\x$ to another profile $\x'$ (if the payment reduction is comparable with the distance between $x_i$ and $x_i'$). More details on the effectiveness of this optimization can be found in Lemma~\ref{lem:opt_outside}.

We highlight that this last optimization is particularly relevant when there is a location $f$ such that when $i$ declares $f$,
then the facility will securely be located in $f$ even if $s_{\ell-1} \neq s_{r+1}$ (this case corresponds to $L_i = R_i = f$). In this case,
it is possible to reduce the cost of every {agent $i$} in profile $\x$ from $b-a$ to $|x_i - f|$.

Nevertheless, we note that the mechanism still has very large costs, namely $b-a$, for at least $\left\lceil \frac{n}{2}\right\rceil - 1$ agents. Indeed, for these agents, it turns out that $L_i = A_i = a$ and $R_i = B_i = b$, and thus $m_i = b-a$. We will show that this is inevitable with a direct-revelation mechanism.

We say that a profile $\x'$ is \emph{$i$-compatible} if $\x'=(\x_{K_i}, \x_{-K_i}')$, i.e., $x_j'=x_j$ for all the agents $j \in K_i$. Before proving the properties of OIM, let us make some observations on $L_i$ and $R_i$.
\begin{lemma}
\label{lem:propLR}
For every $i$, and every $t > L_i$, $f(\x') \geq L_i$ for every $i$-compatible profile $\x'$ with $x'_i = t$,
 and there is one such profile for which $f(\x') = L_i$.
 Similarly, for every $t < R_i$, $f(\x') \leq R_i$ for every $i$-compatible profile $\x'$ with $x'_i = t$,
 and there is one such profile for which $f(\x') = R_i$.
\end{lemma}
\begin{proof}
 Let $t > L_i$ and consider the profile $\x'$ such that
 $$x'_j = \begin{cases}
           x_j, & \text{if } j \in K_i;\\
           t, & \text{if } j = i;\\
           a, & \text{otherwise}.
          \end{cases}
 $$
 It is easy to see that $f(\x') = L_i$. Indeed, if $\ell < 1$,
 then in $\x'$ there are $n-k_i - 1 = \left\lceil \frac{n}{2} \right\rceil - \ell > \left\lceil \frac{n}{2} \right\rceil - 1$ agents whose location is $a$.
 Hence, the leftmost median of $\x'$ must be $a = L_i$.
 If $\ell \geq 1$, then $L_i$ is the $\ell$-th smallest location among agents that are processed before $i$,
 and there are in $\x'$ exactly $n-k_i-1$ agents whose location is surely not larger than $L_i$. Then $L_i$ is the $\ell + n - k_i -1= \left\lceil \frac{n}{2} \right\rceil$-th smallest location in $\x'$, i.e. the (leftmost) median.
 
 On the other hand, it is immediate to see that there is no declaration by agents $j\notin K_i$, with $j \neq i$, that can make
the facility go to the left of $L_i$.
 
 Let now $t < R_i$ and consider the profile $\x'$ such that
 $$x'_j = \begin{cases}
           x_j, & \text{if } j \in K_i;\\
           t, & \text{if } j = i;\\
           b, & \text{otherwise}.
          \end{cases}
 $$
 It is easy to see that $f(\x') = R_i$. Indeed, if $r > k_i$,
 then in $\x'$ there are $n-k_i-1 > \left\lfloor \frac{n}{2} \right\rfloor$ agents whose location is $b$.
 Hence, the median of $\x'$ is $b = R_i$.
 If $r \leq k_i$, then $R_i$ is the $r$-th smallest location among agents that are processed before $i$, and there is in $\x'$ exactly one agent whose location is smaller than $R_i$. Then $R_i$ is the $r + 1 = \left\lceil \frac{n}{2} \right\rceil$-th smallest location in $\x'$, i.e. the (leftmost) median.
 
 Moreover, as above, it is immediate to see that there is no declaration by agents $j\notin K_i$, with $j \neq i$, that can make
 the facility go to the right of $R_i$.
\end{proof}

\begin{lemma}
\label{lem:propLRout}
 For every $i$, and every $t \leq L_i$, $f(\x') \geq t$ for every $i$-compatible profile $\x'$ with $x'_i = t$.
 Similarly, for every $t \geq R_i$, $f(\x') \leq t$ for every $i$-compatible profile $\x'$ with $x'_i = t$.
\end{lemma}
\begin{proof}
 Let $t \leq L_i$ and consider the profile $\x'$ such that
 $$x'_j = \begin{cases}
           x_j, & \text{if } j \in K_i;\\
           t, & \text{if } j = i;\\
           a, & \text{otherwise}.
          \end{cases}
 $$
 If $\ell \geq 1$, then $L_i$ is the $\ell$-th smallest location among agents that are processed before $i$,
 and there are in $\x'$ exactly $n-k_i$ agents whose location is surely not larger than $L_i$. Then $L_i$ is the $\ell + n - k_i = \left\lceil \frac{n}{2} \right\rceil + 1$-th smallest location in $\x'$. Thus, the leftmost median of $\x'$ will be $s_{\ell - 1}$ if $t \leq s_{\ell - 1}$ and $t$ otherwise. When $\ell < 1$, then $L_i=a$ and clearly $f(\x')\geq a$.
 
If $t \geq R_i$, let us consider the profile $\x'$ such that
$$x'_j = \begin{cases}
           x_j, & \text{if } j \in K_i;\\
           t, & \text{if } j = i;\\
           b, & \text{otherwise}.
          \end{cases}
 $$
 If $r \leq k_i$, then $R_i$ is the $r$-th smallest location among agents that are processed before $i$, and there is in $\x'$ no agent whose location is surely smaller than $R_i$. Then $R_i$ is the $r= \left\lceil \frac{n}{2} \right\rceil - 1$-th smallest location in $\x'$. Thus, the leftmost median of $\x'$ will be $s_{r + 1}$ if $t \geq s_{r + 1}$ and $t$ otherwise. When $r > k_i$, then $R_i=b$ and clearly $f(\x')\leq b$.
\end{proof}

We first show that OIM is optimal and OSP with monitoring. Next we will prove that no direct-revelation mechanism enjoys the same properties with lower payments.
\begin{theorem}
\label{thm:algo2osp}
 OIM is an optimal mechanism that is OSP with monitoring.
\end{theorem}
\begin{proof}
We will next prove that the mechanism is OSP with monitoring, and thus for each agent it is obviously dominant to declare her real position. Since OIM places the facility on the median of these positions, it then turns out to be optimal as well.
 
Consider then agent $i$ and let $x_i$ be her real position.
If $s_{\ell - 1} = s_{r + 1} = \lambda$, then the facility will be located in $\lambda$ and $i$ receives a zero payment, regardless of her declaration and the declarations of the agents not in $K_i$.
 
 Suppose now that $s_{\ell - 1} \neq s_{r + 1}$ and the real position of $i$ is $x_i \in [A_i, B_i]$.
 If $i$ declares her real position, then her total cost will be at most $m_i$.
 If $i$ declares a different position $x'_i$, then
for every $i$-compatible profile $\x'$
 \begin{align*}
  c_i(x_i, \M(\x')) &= \max\{d(x_i,f(\x')), d(x'_i,f(\x'))\} - p_i(\x')\\
  & \geq d(x'_i,f(\x')) - p_i(\x') = c_i(x_i', \M(\x')) \geq m_i.
 \end{align*}
 
 Suppose now that $s_{\ell - 1} \neq s_{r + 1}$ and the real position of $i$ is $x_i = A_i - c$ or $x_i  = B_i + c$ with $c > 0$.
 W.l.o.g. we will assume $x_i = A_i - c$.
 If $i$ declares the real position, then her total cost will be at most $m_i+ c$.
  If $i$ declares a position $x'_i = A_i - c'$ or $x'_i = B_i + c'$ with $c' > c$,
 then for every $i$-compatible profile $\x'$
we have that
 \begin{align*}
  c_i(x_i, \M(\x')) &= \max\{d(x_i,f(\x')), d(x'_i,f(\x'))\} - p_i(\x') \geq d(x'_i,f(\x')) - p_i(\x')\\
  & \geq m_i + c' > m_i + c.
 \end{align*}
 
 If $i$ declares a position $x'_i = A_i - c'$ for $0 < c' < c$, then for every $i$-compatible profile $\x'$ we have that
 \begin{align*}
  c_i(x_i, \M(\x')) &= \max\{d(x_i,f(\x')), d(x'_i,f(\x'))\} - p_i(\x') = d(x_i,f(\x')) - p_i(\x')\\
  & = d(x_i, x'_i) + d(x'_i,f(\x')) - p_i(\x') = (c-c') + m_i + c' = m_i + c,
 \end{align*}
 where we used that, according to Lemma~\ref{lem:propLRout}, $f(\x') \geq x_i'$ and thus $d(x_i, f(\x')) = d(x_i, x'_i) + d(x'_i, f(\x))$.
 
 If $i$ instead declares a position $x'_i \in [A_i, B_i+c]$, then for every $i$-compatible profile $\x'$ such that $f(\x') \geq x'_i$ we have that
 \begin{equation}
 \label{eq:OSP_out_left}
 \begin{aligned}
  c_i(x_i, \M(\x')) &= \max\{d(x_i,f(\x')), d(x'_i,f(\x'))\} - p_i(\x') = d(x_i,f(\x')) - p_i(\x')\\
  & = d(x_i, x'_i) + d(x'_i,f(\x')) - p_i(\x') \geq m_i + c.
 \end{aligned}
 \end{equation}
 For every $i$-compatible profile $\x'$ such that $f(\x') < x'_i$, we have instead that
 \begin{equation}
 \label{eq:OSP_out_right}
 \begin{aligned}
  c_i(x_i, \M(\x')) &= \max\{d(x_i,f(\x')), d(x'_i,f(\x'))\} - p_i(\x') = d(x_i,f(\x')) - p_i(\x')\\
  & = d(x_i,f(\x')) - d(x'_i,f(\x')) + m_i + \max\{0, x'_i - B_i\}\\
  & = f(\x') - x_i - x'_i + f(\x') + m_i + \max\{0, x'_i - B_i\}\\
  & \geq (L_i-x_i) - (x'_i - L_i) + m_i  + \max\{0, x'_i - B_i\}\\
  & = (L_i - A_i) + (A_i - x_i) - (x'_i - B_i) - (B_i - L_i) + m_i + \max\{0, x'_i - B_i\}\\
  & \geq m_i + c,
 \end{aligned}
 \end{equation}
 where we used that, according to Lemma~\ref{lem:propLR}, $f(\x') \geq L_i$,
 and that $A_i > a$, and therefore
 \[
 d(x'_i,f(\x')) \leq B_i + c - L_i = L_i - A_i  + c = L_i - x_i \leq d(x_i,f(\x')). \tag*{\qed}
 \]
 \let\qed\relax
\end{proof}

\begin{theorem}
\label{thm:frugal}
  Every optimal OSP direct-revelation mechanism $\M = (f, \p)$
  either sets payments at least as high as OIM,
  or there is an agent $i$ and a profile $\y$ such that $p_i(\y) > 0$.
\end{theorem}
\begin{proof}
 Let $\M = (f,\p)$ be an optimal OSP direct-revelation mechanism,
 and suppose that it assigns the lowest possible payments.
 Fix a player $i$ and recall that $K_i$ is the set of agents whose location is known to $i$ when she is queried.
 Clearly, no mechanism can set lower non-positive payments than OIM  when $s_{\ell-1} = s_{r+1}$.
 Thus, we can safely consider that $s_{\ell-1} \neq s_{r+1}$. 
 Next lemmata show some conditions that payments must satisfy in order for $\M$ to be OSP and to minimize payments.
 Specifically, Lemma~\ref{lemma:in_interval} and Lemma~\ref{lem:in_equal}
 focus on profiles $\x$ such that $x_i \in [A_i, B_i]$ and $L_i \neq R_i$,
 Lemma~\ref{lem:in_equal_limit} consider profiles $\x$ such that $x_i = L_i = R_i$, 
 whereas Lemma~\ref{lemma:out_interval} focuses on the remaining profiles.

 Suppose first that $L_i \neq R_i$.
 Then let $\mu_i$ be the minimum cost that $i$ pays in a profile $\x$ such that $f(\x) \neq x_i$
 (such a profile surely exists since, by optimality of $\M$, $i$ is not a dictator), i.e.,
 $\mu_i = \min_{\x \colon f(\x) \neq x_i} c_i(x_i, \M(\x))$.
 We begin by proving this useful claim.
 \begin{claim}
 \label{claim:start}
  Let $\x$ be a profile such that $f(\x) \neq x_i$ and $c_i(x_i, \M(\x)) = \mu_i$.
  If $x_i < f(\x)$, then for every $y_i \in (x_i, \min\{2f(\x)-x_i,b\}]$, it turns out that
  $c_i(y_i, \M(\y)) = \mu_i$ if $f(\y) \neq y_i$, and $c_i(y_i, \M(\y)) \leq \mu_i$ otherwise.
  
  Similarly, if $x_i > f(\x)$, then for every $y_i \in [\max\{a,2f(\x) - x_i\}, x_i)$, it turns out that
  $c_i(y_i, \M(\y)) = \mu_i$ if $f(\y) \neq y_i$, and $c_i(y_i, \M(\y)) \leq \mu_i$ otherwise.
 \end{claim}
 \begin{proof}
  Since $\M$ is a direct-revelation mechanism, then $i$ diverges on $y_i$ and $x_i$.
  Then, since $\M$ is OSP, it must be the case that
  \begin{align*}
   c_i(y_i, \M(\y)) & \leq c_i(y_i, \M(\x)) = \max \{d(y_i, f(\x)), d(x_i,f(\x))\} - p_i(\x)\\
   & = d(x_i, f(\x)) - p_i(\x) = c_i(x_i, \M(\x)) = \mu_i,
  \end{align*}
  where we used that $d(y_i, f(\x)) \leq d(x_i, f(\x))$ by definition of $y_i$.
  
  However, if $f(\y) \neq y_i$, then, by definition of $\mu_i$, it must be the case that $c_i(y_i, \M(\y)) \geq \mu_i$,
  that leaves $c_i(y_i, \M(\y)) = \mu_i$ as the only possible option.
 \end{proof}
 
 \begin{lemma}
 \label{lemma:in_interval}
  If $L_i \neq R_i$, then for every $y_i \in [A_i, B_i]$, $c_i(y_i, \M(\y)) = \mu_i$ if $f(\y) \neq y_i$,
  and $c_i(y_i, \M(\y)) \leq \mu_i$ otherwise.
 \end{lemma}
 \begin{proof}
  Consider the following procedure:
  \begin{enumerate}
   \item Let $w = 1$, $\Delta^0 = \emptyset$, and $\y^0$ be a profile achieving cost $\mu_i$.
   \item \label{item:loop} Let $t^w = 2f(\y^{w-1})-y^{w-1}_i$.
   \item If $t^w \geq \frac{L_i + R_i}{2}$,
  consider the profile $\y^w$ such that
  $$
   y^w_j = \begin{cases}
            x_j, & \text{if } j \in K_i;\\
            t^w, & \text{if } j = i;\\
            a, & \text{otherwise}.
           \end{cases}
  $$
  and let $\Delta^w = [\max\{a,2f(\y^w) - t^w\}, t^w]$.
  Otherwise consider $\y^w$ such that
  $$
   y^w_j = \begin{cases}
            x_j, & \text{if } j \in K_i;\\
            t^w, & \text{if } j = i;\\
            b, & \text{otherwise}.
           \end{cases}
  $$
  and let $\Delta^w = [t^w, \min\{2 f(\y^w) - t^w,b\}]$.
  \item If $[A_i, B_i] \not\subseteq \Delta^w$, set $w = w+1$ and repeat from step~\ref{item:loop}.
  \end{enumerate}

  Let us first prove, by induction, that for every $w \geq 0$, it holds that $c_i(y^w_i, \M(\y^w)) = \mu_i$.
  This is clearly true for $w = 0$.
  Suppose now that $c_i(y^{w-1}_i, \M(\y^{w-1})) = \mu_i$.
  If $t^w \geq \frac{L_i + R_i}{2}$, then, according to Lemma~\ref{lem:propLR}, we have that $f(\y^w) = L_i \neq y^w_i$.
  Similarly, if $t^w < \frac{L_i + R_i}{2}$, then, according to Lemma~\ref{lem:propLR}, we have that $f(\y^w) = R_i \neq y^w_i$.
  Then, by Claim~\ref{claim:start} applied with $\x = \y^{w-1}$, it holds that $c_i(y^w_i, \M(\y^w)) = \mu_i$.
  In fact, Claim~\ref{claim:start} actually proves that for every $w \geq 0$,
  and every $y_i \in \Delta^w$, it holds that $c_i(y_i, \M(\y)) = \mu_i$ if $f(\y) \neq y_i$,
  and $c_i(y_i, \M(\y)) \leq \mu_i$ otherwise.
  Hence, we are only left to prove that there is a $w$ such that $\Delta^w \supseteq [A_i, B_i]$.
  
  To this aim, we next we prove that for every $w \geq 1$, the size of the range $\Delta^w$ is larger than the size of the range $\Delta^{w-1}$.
  This is clearly true for $w = 1$ since $|\Delta^0| = 0$ and $f(\y^1) \neq t^1$, from which we achieve that $|\Delta^1| \geq 2 |f(\y^1) - t^1| > 0$.
  Consider, instead, $w > 1$.
  Suppose that $\Delta^{w-1} = [\max\{a,2f(\y^{w-1}) - t^{w-1}\}, t^{w-1}]$, from which we have that $|\Delta^{w-1}|\geq 2t^{w-1}-2f(\y^{w-1})$.
  Note that this only occurs if $t^{w-1} \geq \frac{L_i + R_i}{2} = L_i + \frac{R_i - L_i}{2}$ and thus $f(\y^{w-1}) = L_i$,
  from which it follows that $t^w = 2f(\y^{w-1}) - t^{w-1} \leq L_i - \frac{R_i - L_i}{2} < \frac{L_i + R_i}{2}$.
  Therefore $\Delta^w = [t^w, \min\{2 f(\y^w) - t^w,b\}]$.
  If $t_w \geq 2R_i - b$, then $\min\{2 f(\y^w) - t^w,b\} \neq b$,
  and thus $|\Delta_w| = 2 f(\y^w) - 2t^w = 2R_i - 2t_w$,
  otherwise $|\Delta_w| = b - t^w \geq 2R_i - 2t_w$.
  Hence, in both cases we achieve that
  $$|\Delta^w| \geq 2R_i - 2f(\y^{w-1}) + 2t^{w-1} - 2f(\y^{w-1}) = 2t^{w-1} - 2f(\y^{w-1}) + 2|R_i - L_i| > |\Delta^{w-1}|.$$
  The case for $\Delta^{w-1} = [t^{w-1}, \min\{b,2f(\y^{w-1}) - t^{w-1}\}]$ can be similarly proved.
  
  The lemma then follows since the above procedure eventually considers $\Delta \supseteq [A_i, B_i]$.
 \end{proof}
 
 \begin{lemma}
  \label{lem:in_equal}
  If $L_i \neq R_i$, then for every $y_i \in [A_i, B_i]$ and every $i$-compatible profile $\y$ such that $f(\y) = y_i$,
  it holds that $c_i(y_i, \M(\y)) = \mu_i$.
 \end{lemma}
 \begin{proof}
  Suppose that that there is an $i$-compatible profile $\y$ with $f(\y)= y_i$
 and $c_i(y_i, \M(\y)) < \mu_i$.
  Consider $\delta < \min\{\mu_i - c_i(y_i, \M(\y)), \max\{y_i - A_i, B_i - y_i\}\}$.
  According to this choice of $\delta$,  it must exists $t \in [A_i, B_i]$ such that $d(y_i,t) = \delta$.
 Moreover, since $L_i \neq R_i$ by hypothesis, either $t \neq L_i$ or $t \neq R_i$.
  Then, according to Lemma~\ref{lem:propLR}, there is a profile $\y'$ such that $y'_i = t$ and $f(\y') \neq y'_i$.
  Thus, by Lemma~\ref{lemma:in_interval}, it holds that $c_i(y'_i, \M(\y')) = \mu_i$.
  
  However, since $\M$ is a direct-revelation mechanism, then $i$ diverges on $y_i$ and $y'_i$.
  Then, since $\M$ is OSP, it must be the case that
  \begin{align*}
   \mu_i = c_i(y'_i, \M(\y')) & \leq c_i(y'_i, \M(\y)) = \max \{d(y'_i, f(\y)), d(y_i,f(\y))\} - p_i(\y)\\
   & = d(y'_i, y_i) + d(y_i,f(\y)) - p_i(\y) = \delta + c_i(y_i, \M(\y)) < \mu_i,
  \end{align*}
  that is absurd.
 \end{proof}

Consider now the case that $L_i = R_i$.
In this case we let $\mu_i = \min_{\x \colon x_i = L_i} c_i(x_i, \M(\x))$.
Next we prove a lower bound for $\mu_i$.
\begin{lemma}
  \label{lem:in_equal_limit}
   Let $\y'$ be the $i$-compatible profile such that $f(\y') \neq y'_i$ of minimum cost,
i.e., $\y' = \arg \min_{\x \colon f(\x) \neq x_i} c_i(x_i, \M(\x))$. 
  If $L_i = R_i$, then for every $i$-compatible profile $\y$ such that $y_i = L_i$,
  it holds that $c_i(y_i, \M(\y)) \geq c_i(y'_i, \M(\y'))-d(y_i, y'_i)$.
 \end{lemma}
 \begin{proof}
  Suppose that that there is an $i$-compatible profile $\y$ with $y_i=L_i$
 and $c_i(y_i, \M(\y)) < c_i(y'_i, \M(\y'))-d(y_i, y'_i)$. Observe that, since $L_i = R_i$, it must be the case that $f(\y) = y_i$.
  
  Since $\M$ is a direct-revelation mechanism, then $i$ diverges on $y_i$ and $y'_i$.
  Then, since $\M$ is OSP, it must be the case that
  \begin{align*}
   c_i(y'_i, \M(\y')) & \leq c_i(y'_i, \M(\y)) = \max \{d(y'_i, f(\y)), d(y_i,f(\y))\} - p_i(\y)\\
   & = d(y'_i, y_i) + d(y_i,f(\y)) - p_i(\y) < c_i(y'_i, \M(\y')),
  \end{align*}
  that is absurd.
 \end{proof}
 
 \begin{lemma}
 \label{lemma:out_interval}
  For every $c > 0$, if $y_i = A_i-c$, 
then for every $i$-compatible profile $\y$ $c_i(y_i,\M(\y)) \leq \mu_i+c$.

 If $y_i = B_i+c$, 
then for every $i$-compatible profile $\y$ $c_i(y_i,\M(\y)) \leq \mu_i+c$.
%
 \end{lemma}
 \begin{proof}
  Consider $y_i < A_i$.
 and let $\x$ be an $i$-compatible profile with $x_i = A_i$ with $c_i(x_i, \M(\x)) = \mu_i$
 (it exists by Lemma~\ref{lem:in_equal} if $L_i \neq R_i$, and by definition of $\mu_i$, otherwise).
  By definition of $A_i$, we have that $x_i = A_i \leq L_i$.
  Then, by Lemma~\ref{lem:propLRout}, we have that $f(\x) \geq x_i$,
  and thus $d(y_i, f(\x)) = d(y_i, x_i) + d(x_i, f(\x)) = c + d(x_i, f(\x)) > d(x_i,f(\x))$.

  However, since $\M$ is a direct-revelation mechanism, then $i$ diverges on $y_i$ and $x_i$.
  Then, since $\M$ is OSP, it must be the case that
  \begin{align*}
   c(y_i, f(\y)) & \leq c_i(y_i, \M(\x)) = \max \{d(y_i, f(\x)), d(x_i,f(\x))\} - p_i(\x)\\
   & = c + d(x_i,f(\x)) - p_i(\x) = c + \mu_i.
  \end{align*}
%
%
  
  The case for $y_i = B_i + c$ is similar. 
 \end{proof}
 
 These lemmata fix the payments for $i$-compatible profiles $\y$ such that $y_i \in [A_i,B_i]$ when $L_i \neq R_i$.
 As for the remaining cases, next we show how to choose the minimum payments that enable the mechanism $\M$ to be optimal and OSP with monitoring.
 In particular, Lemma~\ref{lem:opt_within} focuses on profiles $\y$ such that $y_i =L_i = R_i$.
 Lemma~\ref{lem:opt_outside} focuses instead on profiles $\y$ with $y_i = A_i - c$ or $y_i = B_i + c$.
 \begin{lemma}
  \label{lem:opt_within}
   Let $\y'$ be the $i$-compatible profile such that $f(\y') \neq y'_i$ of minimum cost,
i.e., $\y' = \arg \min_{\x \colon f(\x) \neq x_i} c_i(x_i, \M(\x))$.
If there is an $i$-compatible profile $\y$ with $y_i = L_i = R_i$ $c_i(y_i, \M(\y)) > c_i(y'_i, \M(\y'))-d(y_i,y'_i)$,
  then there is another direct-revelation optimal OSP mechanism $\M'$
  that for each profile assigns payments at least as small as $\M$ and for at least one profile it assigns a smaller payment.
 \end{lemma}
 \begin{proof}
  Consider $\M' = (f, \p')$ as $\M$ except that it sets payments such that 
$c'_i(y_i, \M'(\y)) = d(y_i,f(\y)) - p'_i(\y) = c_i(y'_i, \M(\y'))-d(y_i,y'_i)$, and for every $\x'$ such that $x'_i \neq y_i$,
it sets payments such that $c'_i(x'_i, \M'(\x')) - \min_{\x \colon x_i = y_i} c'_i(x_i, \M'(\x)) = c_i(x'_i, \M(\x')) - \min_{\x \colon x_i = y_i} c_i(x_i, \M(\x))$.
  Clearly, $\M'$ is a direct-revelation mechanism. Moreover, since it places the facility in the median location,
  it is optimal if it is OSP. Finally, $\M'$ reduces the payment assigned to $i$ at least in the profile $\y$.
  
  Hence, it is only left to show that $\M'$ is OSP.
  Clearly, the OSP condition still holds between two profiles in which the location of $i$ is different from $y_i$,
  and when the real location of $i$ is exactly $y_i$. Next we show, that if the real location of $i$ is $x'_i \neq y_i$,
  then it is not convenient for $i$ to declare $y_i$. That is, we prove that $c'_i(x'_i, \M'(\x')) \leq c'_i(x'_i, \M'(\y))$ for every $\x'$.
  
  
  Since $L_i = R_i$, it must be the case that $f(\y)=y_i=A_i = B_i$, and then either $x'_i = A_i - c$ or $x'_i = B_i + c$, with $c > 0$.
  According to Lemma~\ref{lemma:out_interval}, we have that $c'_i(x'_i, \M'(\x')) - c'_i(y_i, \M'(\y)) = c_i(x'_i, \M(\x')) - \min_{\x \colon x_i = y_i} c_i(x_i, \M(\x)) \leq c$.
  Instead, 
  \begin{align*}
   c'_i(x_i, \M'(\y)) & = \max\{d(x_i, f(\y)), d(y_i, f(\y))\} + p_i(\y) = d(x_i, f(\y)) - p_i(\y)\\
   & = d(x_i, y_i) + c'_i(y_i,\M'(\y)) = c + c'_i(y_i,\M'(\y)). \tag*{\qed}
  \end{align*}
\let\qed\relax
 \end{proof}

 \begin{lemma}
 \label{lem:opt_outside}
  If there is $\y$ with $y_i = A_i - c$, $c > 0$, and $c_i(y_i, \M(\y)) < \mu_i+c$,
 then there is another direct-revelation optimal OSP mechanism $\M'$
  that for each profile assigns payments at least as small as $\M$ and for at least one profile it assigns a smaller payment.
  
  Similarly, if there is $\y$ with $y_i = B_i + c$ and $c_i(y_i, \M(\y)) < \mu_i+c$,
 then there is another direct-revelation optimal OSP mechanism $\M'$
  that for each profile assigns payments at least as small as $\M$ and for at least one profile it assigns a smaller payment.
 \end{lemma}
 \begin{proof}
  By Lemma~\ref{lemma:in_interval}, Lemma~\ref{lem:in_equal}, and Lemma~\ref{lem:opt_within}, we know that $c_i(y_i, \M(\y)) = \mu_i$ for every $\y$ such that $y_i \in [A_i,B_i]$.
  
  Let $y^A_i = \max_{c > 0} \{y_i = A_i-c \colon c_i(y_i, \M(\y)) < \mu_i+c\}$
  and $y^B_i = \min_{c > 0} \{y_i = B_i+c \colon c_i(y_i, \M(\y)) < \mu_i+c\}$.
  Denote as $y^*_i$ the one closer to the interval $[A_i, B_i]$,
  i.e. $y^*_i = \arg \min_{y_i = y^A_i, y^B_i} \min \{d(A_i, y_i), d(B_i, y_i)\}$.
  Henceforth, we assume w.l.o.g. that $y^*_i = y^A_i$.
  Let also $\y^* = \min_{\hat{\y} \colon \hat{y}_i = y^*_i} c_i(\hat{y}_i, \M(\hat{\y}))$.
  Observe that, by definition, it must be the case that $c_i(y^*_i, \M(\y^*)) < \mu_i+c^*$,
 where $c^* = d(y^*_i,A_i)$.
  
  Consider $\M' = (f, \p')$ as follows: for every  $\y$ such that $y_i \in (A_i - c^*, B_i + c^*)$,
  set payments such that $c'_i(y_i, \M'(\y)) = d(y_i, f(\y)) + p'_i(\y) = c_i(y_i, \M(\y)) - (\mu_i+c^*-c_i(y^*_i, \M(\y^*))$;
  for every $\y$ such that $y_i \in \{y^*_i,B_i+c^*\}$, set payments such that $c'_i(y_i, \M'(\y)) = c_i(y^*_i, \M(\y^*))$;
  for every remaining profile $\y$, set payments such that $c'_i(y_i, \M'(\y)) = c_i(y_i, \M(\y))$.
  
  Clearly, $p'_i(\y) < p_i(\y)$ for every $\y$ such that $y_i \in (A_i - c^*, B_i + c^*)$,
  whereas $p'_i(\y) \leq p_i(\y)$ for every other profile $\y$.
  Moreover, $\M'$ is a direct-revelation mechanism, and, since it places the facility in the median location,
  it is optimal if it is OSP. Hence, it is only left to show that $\M'$ is OSP.
  
  It is immediate to see that the OSP condition holds if
  from every profile one moves to another profile in which the location of $i$ is $x_i < y^*_i$ or $x_i > B_i+c^*$, and if
  from a profile in which the location of $i$ is $x_i \in (A_i - c^*, B_i + c^*)$
  one moves to another profile in which the location of $i$ is $y_i \in (A_i - c^*, B_i + c^*)$.
  Next we show that even if the real location of $i$ is $x_i \leq y^*_i$ or $x_i \geq B_i+c^*$,
  then it is not convenient for $i$ to declare $y_i \in(A_i - c^*, B_i + c^*)$.
  That is, we prove that $c'_i(x_i, \M'(\x)) \leq c'_i(x_i, \M'(\y))$ for every $\x$ and $\y$ such that $x_i$ and $y_i$ are as above.

  To this aim, let $\mu'_i = c_i(y^*_i, \M(\y^*) - c^*$.
    By construction, 
  $c'_i(x_i, \M'(\x))  \mu'_i$ when $x_i \in [A_i,B_i]$,
  $c'_i(x_i, \M'(\x)) = \mu'_i+c$ when $x_i = A_i - c$ or $x_i = B_i + c$, with $0 \leq c \leq c^*$,
  whereas $c'_i(x_i, \M'(\x)) \leq \mu'_i+c$ when $x_i = A_i - c$ or $x_i = B_i + c$, with $c > c^*$.
  Then, the OSP condition can be proved as in \eqref{eq:OSP_out_left} and \eqref{eq:OSP_out_right}
  with $\mu'_i$ in place of $m_i$.
 \end{proof}

 In conclusion, a direct-revelation mechanism $\M$ that is OSP, optimal and whose payments cannot be lowered must be exactly as  OIM with $\mu_i$ in place of $m_i$. However, suppose w.l.o.g. that $m_i = R_i - A_i$ and consider a profile $\x$ such that $x_i = A_i$ and $f(\x) = R_i$.
 According to Lemma~\ref{lem:propLR}, such a profile surely exists. Moreover, as showed above, $c_i(x_i, \M(\x)) = \mu_i$.
 Therefore, if all payments are non-positive, then
 \[
  \mu_i = c_i(x_i, \M(\x)) = d(x_i, f(\x)) - p_i(\x) \geq d(x_i, f(\x)) = R_i - A_i = m_i. \tag*{\qed}
 \]
 \let\qed\relax
\end{proof}

 Theorem~\ref{thm:frugal} does not rule out the existence of a \emph{non-direct-revelation} mechanism that is OSP with lower payments.
 However, we next show that for \emph{every} mechanism there is at least one instance
 on which it cannot set payments lower than the one assigned by OIM.
 In particular, it turns out that for every mechanism there is at least one instance
 on which at least $\left\lceil \frac{n}{2}\right\rceil - 1$ agents incur in a very large cost, namely $b-a$.

\begin{lemma}
\label{lem:frugal2}
 For every optimal OSP mechanism with monitoring there is an instance of the facility location problem
 for which the mechanism sets payments at least as high as OIM.
\end{lemma}
\begin{proof}
 Let $\M = (f,\p)$ be an optimal OSP mechanism and let $\delta$ be an arbitrary constant.
 For every player $i$, we define $t_\M(i)$ as the time in which during the execution of mechanism $\M$ agent $i$ diverges on types $a$ and $b-\delta$.
 That is, $t_\M(i)$ is the time step in which $\M$ asks agent $i$ to take an action when her type is $a$ that is different
 from the action she takes if her type is $b-\delta$.
 Note that $\M$ may ask to more than one agent to diverge on types $a$ and $b-\delta$ at the same time.
 Moreover, there may be agents that never diverge on types $a$ and $b-\delta$ (for which we set $t_\M(i) = \infty$).
 However, these are at most $\frac{n-1}{2}$, otherwise $\M$ must give the same output on instance $\x=(b-\delta, \ldots, b-\delta)$
 and on the instance $\y$ such that $y_j = b-\delta$ if $t_\M(j) < \infty$, and $y_j = a$ otherwise.
 But this contradicts the optimality of $\M$.
 
 We now show that for the following instance the payments assigned by $\M$ are at least as high as the payment assigned by OIM. We assume that the number $n$ of agents is odd, and, without loss of generality, that agents are labeled so that
 $t_\M(1) \leq t_\M(2) \leq \ldots \leq t_\M(n)$.
 We consider the instance $\x$ according to which
 the real position of agent $i$ is $x_i = b-\delta$ if $i \leq \frac{n+1}{2}$,
 and $x_j = a$, otherwise.
 It is not hard to see that if agents are processed in increasing order of their label,
 then OIM on this instance assigns a zero payment to every agent whose real position is $a$,
 and, among agents with real position $b-\delta$, only to the last to be queried.
 It assigns instead a payment of $b-a$ to every remaining agent.
 
 We next show that for $j \leq \frac{n-1}{2}$, it must be the case that
 $\M$ also sets $p_j^{\M}(\x) \leq \delta - (b-a)$.
 Let indeed $t'$ be the first step within mechanism $\M$ in which $j$ diverges on types $b-\delta$ and $b$.
 We set $t^* = \min\{t', t_\M(j)\}$.
 Consider then the following instance $\y$:
 $y_j = b$,
 $y_k = b-\delta$ if $t_\M(k) < t^*$ and $k \neq j$,
 and $y_k = a$ otherwise.
 Note that there are at most $j-1 \leq \frac{n-1}{2} - 1$ agents whose location is $b-\delta$
 and at least $\frac{n+1}{2}$ agents whose location is $a$.
 Then, by optimality of $\M$, we have that $f(\y) = a$.
 Moreover, since $\M$ is OSP, it follows that
 \begin{align*}
  & d(y_j, f(\y)) - p_j^\M(\y) \leq \max\{d(y_j,f(\x)),d(x_j,f(\x))\} - p_j^\M(\x)\\
  \Rightarrow \qquad & p_j^\M(\x) \leq \delta - (b - a) + p_j^\M(\y).
 \end{align*}
 Since $p_j^\M(\y) \leq 0$, the claim then follows by having $\delta$ going to 0.
 \end{proof}
 
The lemma above does not exclude that an indirect optimal mechanism might set payments to the agents smaller than OIM's only for some specific order in which agents are queried. It is left open to understand if this the case.
However, we remark that OIM maintains OSP irrespectively of such an ordering.
 
\section{Conclusions}
We have studied the limitations of OSP mechanisms in terms of the approximation guarantee of their outputs. By focusing on two paradigmatic problems in the literature, machine scheduling and facility location, we have shown that OSP can yield a significant loss in the quality of the solutions returned. We have proposed the use of a novel mechanism design paradigm, namely monitoring, as a way to reconcile OSP with good approximations. Our positive results show how the ingredients needed for truthfulness with monitoring marry up the demands needed for OSP. 

We leave open the problem of understanding the extent to which this parallel holds in general. Several additional open problems pertain the two case studies considered. For machine scheduling, it would be interesting to see whether the lower bound can be improved. For facility location, it is interesting to establish if indirect mechanisms can be more frugal for the agents. More generally, the mechanisms with monitoring for which we provide an OSP implementation are shown to be collusion-resistant; is there any way to guarantee OSP (with monitoring) without relying on coalitional notions of incentive-compatibility? And how hard is it to design OSP mechanisms that do not use any additional control on agents' declarations?

\bibliographystyle{plainnat}
\bibliography{osp}
\end{document}